\documentclass[conference, 10pt]{IEEEtran}
\usepackage{setspace}
\usepackage{amsmath}
\usepackage{amssymb}
\usepackage{mathrsfs}
\usepackage{amsthm}
\usepackage{multirow}
\usepackage{color}
\usepackage{array}
\usepackage{graphicx}
\usepackage{float}
\usepackage{url}
\usepackage{comment}
\usepackage{enumerate}
\usepackage{enumitem}
\usepackage{eucal}
\usepackage{tabularx}
\usepackage{makecell}
\usepackage{cleveref}
\usepackage{fixfoot}
\usepackage{caption}
\renewcommand{\thetable}{\arabic{table}}

\usepackage[style=ieee]{biblatex}
\usepackage[noend]{algorithmic}
\usepackage[table]{xcolor}
\usepackage{blkarray}

\usepackage{stmaryrd}
\usepackage{balance}

\usepackage[top=0.55in, bottom=0.97in, left=0.6in, right=0.62in]{geometry}

\IEEEoverridecommandlockouts

\theoremstyle{plain}
\newtheorem{theorem}{Theorem}
\newtheorem{corollary}[theorem]{Corollary}

\theoremstyle{definition}
\newtheorem{definition}{Definition}

\newtheorem{remark}[definition]{Remark}

\newcommand{\FF}{\mathbb{F}}

\DeclareMathAlphabet{\mathbfsl}{OT1}{ppl}{b}{it} 

\usepackage{algorithm}
\floatname{algorithm}{Scheme}

\newcolumntype{Y}{>{\centering\arraybackslash}X} 

\title{$k$-server Byzantine-Resistant PIR Scheme with Optimal Download Rate and Optimal File Size\\[-3mm]}
\author{
  \IEEEauthorblockN{Stanislav Kruglik\IEEEauthorrefmark{1}, Son Hoang Dau\IEEEauthorrefmark{2}, Han Mao Kiah\IEEEauthorrefmark{1}, Huaxiong Wang\IEEEauthorrefmark{1}}
	
 \IEEEauthorblockA{\small \IEEEauthorrefmark{1} School of Physical and Mathematical Sciences, 
		Nanyang Technological University, Singapore
    }
 \IEEEauthorblockA{\small \IEEEauthorrefmark{2} School of Computing Technologies, STEM College, RMIT University, Melbourne, Australia
    }
  {\small stanislav.kruglik@ntu.edu.sg, sonhoang.dau@rmit.edu.au, hmkiah@ntu.edu.sg, hxwang@ntu.edu.sg}}

\begin{document}
\date{}

\maketitle

\hspace*{-12pt}
\begin{abstract}
We consider the problem of designing a Private Information Retrieval (PIR) scheme on $m$ files replicated on $k$ servers that can collude or, even worse, can return incorrect answers. Our goal is to correctly retrieve a specific message while keeping its identity private from the database servers. We consider the asymptotic information-theoretic capacity of this problem defined as the maximum ratio of the number of correctly retrieved symbols to the downloaded one for a large enough number of stored files. We propose an achievable scheme with a small file size and prove that such a file size is minimal \textcolor{black}{for the fixed number of retrieved symbols}, solving the problem pointed out by Banawan and Ulukus. 
\end{abstract}

\setstretch{0.97}

\section{Introduction}
A Private Information Retrieval (PIR) scheme is a tool to retrieve a given file $f_{\iota}$ from a database $\mathbf{x}=(f^{(1)},\ldots, f^{(m)})$, while keeping its identity $\iota\in[m]$ private for the database servers \cite{overview, overview2}. The setup has many related practical applications, including protecting the identity of stock market records reviewed by investment funds because showing specific interest may negatively affect the stock price. The first PIR scheme was proposed in the pioneering paper by Chor et. al. \cite{chor}. In the case of a single server, the authors also showed that to guarantee information-theoretical privacy of retrieved file index, it is necessary for the user to download the entire database.
Thus, to reduce the communication cost in an information-theoretical setting, we have to move to a multi-server setup. In this model, the client queries each of $k$ servers once, while keeping the identity of the retrieved file private from up to $t$ honest-but-curious servers. 
In PIR literature, such a scheme is called $t$-private $k$-server PIR scheme and such property is known as $t$-privacy.

The computer science formulation of the PIR problem assumes files of size one and measures the performance by the sum of the lengths of queries (upload cost) and the sum of the length of responses (download cost)~\cite{Yekhanin, Yekhanin2, Dvir}. Motivated by practical applications, in which the size of the message can be arbitrarily large, the problem of PIR was revisited by the information-theory community. Download cost became the dominant performance metric, and the maximum achievable download rate, defined as a ratio of the retrieved file size to the amount of information downloaded by the user, became a focus of the pleiad of the research papers~\cite{dr1, dr2, dr3, SunJafar}. 

Most of the current PIR schemes assume that servers are honest-but-curious and provide correct answers. However, such an assumption cannot be guaranteed in the cloud environment. This fact poses an interesting question about responses to wrong server answers. Here we provide three different interpretation of this question and their formal definition.

\begin{itemize}
\item \textit{$s$-verifiability [referred as $s$-security in \cite{zhangwang}]}. The client can detect the presence of up to $s$ servers that persuade the client to output a wrong result (\cite{zhangwang, zhang, kruglik2023, cao2023}). 
\item \textit{$a$-accountability}. The client can identify each of up to $a$ servers that persuade the client to output a wrong result (\cite{zhang2, zhao}).
\item \textit{$b$-byzantine resistance/$b$-byzantine robustness}. The client can retrieve the correct result in presence of up to $b$ servers that persuade the client to output a wrong result (\cite{BPIR1, BPIR2, BPIR3, BPIR4}).
\end{itemize}

It is clear that $a$-accountability implies $a$-verifiability, while $b$-byzantine resistance implies both $b$-accountability and $b$-verifiability~\cite{zhangwang}. However, in some practical applications, the user needs to be able to correctly reconstruct the desired message irrespective of the adversarial actions of servers. This fact motivates us to consider the strongest notion of $b$-byzantine resistant PIR~\cite{BPIR1, BPIR2, BPIR3, BPIR4}. 

The capacity of $b$-byzantine resistant PIR scheme for the case of $2b+t<k$ is shown in \cite{BPIR1} to be equal to
\begin{equation}
    C_m(t,b,k)=\frac{k-2b}{k}\cdot\frac{1-\frac{t}{k-2b}}{1-\left(\frac{t}{k-2b}\right)^m}.
\end{equation}

Authors of \cite{BPIR1} also proposed a general achievable scheme based on MDS codes. It utilizes the division of each file into multiple sub-packets whose number is denoted as sub-packetization. 
In~\cite{BPIR1}, the scheme has sub-packetization value $(k-2b)^m$, while the problem of obtaining its minimum capacity-achieving quantity is left as an open one. We do note that each sub-packet usually corresponds to some finite field element, and the size of the latter drastically affects the implementation costs~\cite{implementation}. Thus, in this paper, we focus on the total file size. 
Since the number of files is high, we are interested in asymptotic capacity values, where the size of the file in scheme from~\cite{BPIR1} is tremendous. So, we let $m\to \infty$, and for the case of $2b+t<k$ we have 
\begin{equation}\label{ascapacity}
C(t,b,k)\triangleq\lim_{m\to\infty} C_m(t,b, k)=\frac{k-2b}{k}\cdot\left(1-\frac{t}{k-2b}\right).
\end{equation}

\textcolor{black}{There has been considerable research on reducing sub-packetization levels and file sizes for different PIR setups, including $1$-colluding replicated PIR~\cite{chao2019}, $1$-colluding MDS-coded PIR~\cite{xu2018, zhu2020, zhou2020} and $t$-colluding replicated PIR~\cite{zhang2019}. However, to the best of our knowledge, there are no papers that consider a similar problem for Byzantine-resistant PIR. To close this gap, in this paper, we propose non-universal $b$-byzantine resistant $k$-server PIR with optimal download rate and small file size for asymptotically large number of files.} We also formally prove that the latter is minimal among all capacity-achieving schemes. The key ingredients of our method are the recently proposed communication-efficient secret sharing scheme based on trace recovery framework \cite{communicationefficientSS} and the technique to repair Reed-Solomon code in presence of erroneous traces  \cite{barg}.

\section{Preliminaries}
\subsection{Notations}
For any integer $n>0$ we denote $[n]=\{1,\ldots,n\}$. For any prime power $q$, we denote an extended finite field with $q^s$ elements as $\mathbb{F}_{q^s}$. The base field with $q$ elements is denoted as $\mathbb{F}_{q}$. For any $\xi\in\mathbb{F}_{q^s}$ we define the trace function from $\mathbb{F}_{q^s}$ to $\mathbb{F}_{q}$ as $\textrm{Tr}(\xi)=\sum_{i=0}^{s-1}\xi^{q^i}$. We note that it is $\mathbb{F}_{q}$-linear function. By $\mathbb{F}_{q^s}[\xi]$ we denote the ring of polynomials over $\mathbb{F}_{q^s}$. By superscript
$T$, we denote the transpose of a vector. By $\mathbf{M}_{ij}$ we denote $(i, j)$th entry of matrix $\mathbf{M}$. By $\langle \mathbf{M}, \mathbf{N} \rangle$ we denote the Frobenius inner product of $\mathbf{M}$ and $\mathbf{N}$, 
i.e. $\langle \mathbf{M}, \mathbf{N} \rangle = \sum_{i,j} M_{ij}N_{ij}$. \textcolor{black}{By $H(X)$ - we denote the entropy of discrete random variable $X$.} 

\subsection{$k$-Server PIR Schemes}
Let us formally define $k$-server PIR schemes. Let the database $\mathbf{x}$ be formed of $m$ files $f^{(1)},\ldots,f^{(m)}$ and replicated on each server. The user wants to retrieve the file $\iota$ by sending the queries $\textbf{q}_1,\ldots,\textbf{q}_k$ to each server.
Based on the received query $\textbf{q}_j$, each server $j\in[k]$ computes the answer $\textbf{a}_j$ and sends it back to the user. 
\textcolor{black}{In the byzantine PIR setting, there exists unknown to the user set of up to $b$ servers that can provide incorrect answers to queries. After this introduction, we can define $k$-server $t$-private $b$-byzantine resistant PIR.}
\begin{definition}[$k$-server $t$-private $b$-byzantine resistant PIR]\label{PIRdef} A $k$-server $t$-private $b$-byzantine resistant PIR is a scheme that satisfies the following properties:
\begin{enumerate}
    \item (\textit{Privacy}) \textcolor{black}{The scheme is $t$-private, i.e., any subset of $t$ or less queries do not reveal any information about the identity of the file.} 
    \item (\textit{Correctness}) \textcolor{black}{The scheme is correct and $b$-byzantine resistant, i.e., the user is always able to successfully decode the file from any $k$ queries and corresponding answers even if $b$ answers are incorrect. We note that the set of $b$ incorrect responses {\em a priori} is not known to the user.}
    
\end{enumerate}
\end{definition}
\begin{remark}
By setting $b=0$ this definition is \textcolor{black}{reduced} to $k$-server $t$-private PIR scheme. 
\end{remark}
\begin{definition}[retrieval threshold]
\textcolor{black}{A $k$-server $t$-private $b$-byzantine resistant PIR scheme from Definition~\ref{PIRdef} has the retrieval threshold $r$ if, for all sets of $r$ and more answers, the user is always able to successfully decode the file from these answers and corresponding queries, even if $b$ answers are incorrect. As before, we note that the set of $b$ incorrect responses {\em a priori} is not known to the user.}
\end{definition}

\subsection{A Communication-Efficient PIR Scheme}\label{basicscheme}

Let us adopt a communication-efficient secret-sharing scheme from \cite{communicationefficientSS} to obtain $k$-server $t$-private PIR scheme with optimal download rate. For simplicity, we consider a non-universal case when we request responses from exactly $k\geq r$ servers, where $r$ is the recovery threshold and \textcolor{black}{$(r-t)$ divides $(k-t)$. In the same way as in the Reed-Solomon repairing problem, we can reduce the total download cost by increasing the number of servers involved~\cite{guruswami2017}.} 

\noindent\rule{9cm}{0.8pt}
\textit{Scheme\;$\Pi_1$:}\;$k$-server\;$t$-private\;PIR
\noindent\rule{9cm}{0.4pt}

Let $t$, $r$, $k$ be positive integers satisfying $t<r<k\leq q$, $\Delta=r-t$ and $\Delta|(k-t)$.  Denote by $s\triangleq\frac{k-t}{\Delta}$. Let $\Omega_{\alpha}=\{\alpha_1,\ldots,\alpha_{\Delta}\}\subset\mathbb{F}_{q^s}$, $\Omega_{\chi}=\{\chi_1,\ldots,\chi_t\}\subset\mathbb{F}_{q^s}$ and $\Omega_{\beta}=\{\beta_1,\ldots,\beta_k\}\subseteq\mathbb{F}_{q}$ be publicly known non-intersecting sets such that all elements of $\Omega_{\alpha}$ are roots of distinct monic irreducible polynomials of degree $s$ over $\mathbb{F}_q$.

Let us represent the database $\mathbf{x}$ with $m$ files as a $m\times\Delta$-array. Let the $(i,j)$-th entry of $\mathbf{x}$ be $x^{(i)}_j$. Then we set the file $f^{(i)}\triangleq[x^{(i)}_1,\ldots, x^{(i)}_\Delta]$.
Therefore,
\[
\mathbf{x} = \begin{bmatrix}
	x^{(1)}_1 & x^{(1)}_2 & \cdots & x^{(1)}_\Delta\\
	x^{(2)}_1 & x^{(2)}_2 & \cdots & x^{(2)}_\Delta\\
	\vdots & \vdots & \ddots &\vdots\\
	x^{(m)}_1 & x^{(m)}_2 & \cdots & x^{(m)}_\Delta
\end{bmatrix}	= 
\begin{bmatrix}
	f^{(1)} \\
	f^{(2)} \\
	\vdots \\
	f^{(m)}
\end{bmatrix}\,.
\] 
Then we define $\mathbf{e}_{(i,j)}$ to be the $m\times \Delta$ indicator array for the database components. 
In other words,  the $(i,j)$th entry of $\mathbf{e}_{(i,j)}$ is one, while all other entries of $\mathbf{e}_{(i,j)}$ are zero. We replicate the database $\mathbf{x}$ on $k$ servers.
\begin{itemize}
    \item \textit{Query generation algorithm:} To retrieve the file $\iota\in[m]$ user randomly generates $t$ $(m\times\Delta)$-arrays $\mathbf{r}^{(1)},\ldots,\mathbf{r}^{(t)}$ and draw a random degree-$(t+\Delta-1)$ curve
    \begin{align}
\mathbf{g}(\xi) 
& \triangleq \sum_{j=1}^\Delta \prod_{\ell\in[\Delta]\setminus\{j\}}\left(\frac{\xi-\alpha_\ell}{\alpha_j-\alpha_\ell}\right)
							\prod_{\ell=1}^t \left(\frac{\xi-\chi_\ell}{\alpha_j-\chi_\ell}\right) \mathbf{e}_{(\iota,j)} \notag \\
& + \sum_{h=1}^t\prod_{\ell=1}^\Delta \left(\frac{\xi-\alpha_\ell}{\chi_h-\alpha_\ell}\right)
\prod_{\ell\in[t]\setminus\{h\}} \left(\frac{\xi-\chi_\ell}{\chi_h-\chi_\ell}\right) \mathbf{r}^{(h)}\label{eq:query}
\end{align}
that resides in $\mathbb{F}_{q^s}^{m\times\Delta}$ and passes through points $(\alpha_1,\mathbf{e}_{(\iota,1)}),\ldots, (\alpha_{\Delta},\mathbf{e}_{(\iota,\Delta)})$. Query to server $j\in[k]$ is $\mathbf{g}(\beta_j)$. We note that both $\mathbf{g}$ and $\mathbf{r}^{(h)}$ depend on retrieved index $\iota$, but we omit the subscript $\iota$ for readability.

\item \textit{Answer generation algorithm:} Upon receive the query $\mathbf{g}(\beta_j)$, server $j\in[k]$ computes the Frobenius inner product $\langle \mathbf{g}(\beta_j), \mathbf{x} \rangle$. We can observe that 
\begin{align*}
&\langle \mathbf{g}(\xi), \mathbf{x} \rangle 
 = \sum_{j=1}^\Delta \prod_{\ell\in [\Delta]\ne\{j\}}\left(\frac{\xi-\alpha_\ell}{\alpha_j-\alpha_\ell}\right)
\prod_{\ell=1}^t \left(\frac{\xi-\chi_\ell}{\alpha_j-\chi_\ell}\right) x^{(\iota)}_j \notag \\
& + \sum_{h=1}^t\prod_{\ell=1}^\Delta \left(\frac{\xi-\alpha_\ell}{\chi_h-\alpha_\ell}\right)
\prod_{\ell\in[t]\setminus\{h\}} \left(\frac{\xi-\chi_\ell}{\chi_h-\chi_\ell}\right) \langle\mathbf{r}^{(h)},\mathbf{x} \rangle,
\end{align*}
which is a polynomial in $\xi$ of degree $\Delta+t-1=r-1$. 
We call this polynomial $\phi(\xi)$ and observe further that $\phi(\alpha_i)=x^{(\iota)}_i$ for $i\in[\Delta]$.
\begin{itemize}
    \item For retrieval from answers from $r$ servers, server $j$ responds with value of $\mathbf{a}_j=\phi(\beta_j)\in \mathbb{F}_{q^s}$.
    \item For retrieval from answers from $k$ servers, server responds with 
    \begin{equation}
    \mathbf{a}_j=\textrm{Tr}(v_j\phi(\beta_j)) \in \mathbb{F}_q,   \end{equation}
    where 
    \begin{equation}
    v_j=\prod_{\ell=1}^{\Delta}(\beta_j-\alpha_{\ell})^{-1}\times\prod_{\ell\in[k]\setminus\{j\}}(\beta_j-\beta_\ell)^{-1}.
\end{equation}
\end{itemize}
\item \textit{File retrieval algorithm} 
\begin{itemize}
    \item For retrieval from answers from $r$ servers, the user applies the Lagrange interpolation formula.
    \item For retrieval from answers from $k$ servers, the user prepares a basis $\{\theta_1,\ldots,\theta_s\}$ for $\mathbb{F}_{q^s}$ over $\mathbb{F}_q$ and its trace-orthogonal basis $\{\eta_1,\ldots,\eta_s\}$. After that, the user chooses polynomials $h_{i\delta}\in\mathbb{F}_q[\xi]$ of degree less than $s$ for all $i\in[\Delta]$ and $\delta\in[s]$ so that 
    \begin{equation}
    h_{i\delta} (\alpha_{i})=u_{i}^{-1}\eta_{\delta}\prod_{\ell\in[\Delta]\setminus\{i\}}\tilde{f}_{\ell}^{-1}(\alpha_{i}),
    \end{equation}
    where $\tilde{f}_{\ell}(\xi)$ is the minimal polynomial of $\alpha_{\ell}$ over $\mathbb{F}_q$ and 
    \begin{equation}
    u_{i}=\prod_{\ell\in[\Delta]\setminus\{i\}}(\alpha_{i}-\alpha_{\ell})^{-1}\times\prod_{j=1}^k(\alpha_{i}-\beta_j)^{-1}. 
    \end{equation}
    The user retrieves the file of interest by
    \begin{align}
    &x_{i}^{(\iota)}=\phi(\alpha_{i})=-\sum_{\delta=1}^s\theta_{\delta}\Bigg(\sum_{j=1}^kh_{i\delta}(\beta_j)\textrm{Tr}(v_j\phi(\beta_j))\cdot\notag \\
    &\prod_{\ell\in[\Delta]\setminus\{i\}}\tilde{f}_{\ell}(\beta_j)\Bigg),
    \end{align}
    for all $i\in[\Delta]$.
\end{itemize}
\end{itemize}
\noindent\rule{9cm}{0.4pt}
\begin{theorem}\label{theorem1}
Scheme~$\Pi_1$ is $k$-server $t$-private PIR over $\mathbb{F}_{q^s}$ with file size $(k-t)\log(q)$ and a recovery threshold $r$ that achieves the asymptotic capacity~\eqref{ascapacity} for $b=0$ and any given $\Delta$ and $r$ so that
\begin{equation*}
t<r<k\leq q,\;\Delta=r-t\;\;\textrm{and}\;\;\Delta|(k-t), 
\end{equation*}
and $s=\frac{k-t}{\Delta}$.
\end{theorem}
\begin{proof}
According to the definition of $k$-server $t$-private PIR, we will prove privacy and correctness properties and show that responses from $r$ servers are enough for file retrieval. The proof is very similar to the proof from \cite{communicationefficientSS}. To make the paper self-contained, we present the proof here in all the details.

To prove the security, we need to show that 
\begin{equation}\label{sec1}
I(\mathbf{g}(\beta_{l_1}),\ldots, \mathbf{g}(\beta_{l_t});\mathbf{e}_{\iota,1},\ldots,\mathbf{e}_{\iota, \Delta})=0,
\end{equation}
for any subset $\{l_1,\ldots,l_t\}\subset[k]$ of servers and any file index $\iota\in[m]$.

As each element of the matrix $\mathbf{g}(\xi)$ is encoded separately from other elements and corresponding random symbols are independent, $(i,j)$th entry of $\mathbf{g}(\xi)$ depends only on $\mathbf{e}_{(i,j)}$ and conditionally independent of everything else. Hence, our scheme is equivalent to the transmission over $m\Delta$ independent channels \cite{cover} and, as a result, we have
\begin{align}\label{sec2}
&I(\mathbf{g}(\beta_{l_1}),\ldots, \mathbf{g}(\beta_{l_t});\mathbf{e}_{\iota,1},\ldots,\mathbf{e}_{\iota, \Delta})\notag\\ 
&\leq \sum_{i=1}^m\sum_{j=1}^{\Delta} I(\mathbf{g}(\beta_{l_1})_{(ij)},\ldots, \mathbf{g}(\beta_{l_t})_{(ij)};(\mathbf{e}_{\iota,1})_{(ij)},\ldots,(\mathbf{e}_{\iota, \Delta})_{(ij)}).
\end{align}

It can be easily seen that for each $i, j$, $\mathbf{g}(\beta_{l_1})_{(ij)},\ldots,\mathbf{g}(\beta_{l_t})_{(ij)}$ are $t$ evaluations of random polynomial $\psi_{(ij)}$ of degree $t+\Delta-1$ over $\mathbb{F}_{q^s}$ at \textcolor{black}{$l_t$ different} points $\beta_{l_1},\ldots,\beta_{l_t}$. \textcolor{black}{Hence, for} any given values of $(\mathbf{e}_{\iota,1})_{ij},\ldots,(\mathbf{e}_{\iota, \Delta})_{ij}$ by Lagrange interpolating formula we can obtain a unique polynomial $\psi_{(ij)}$ over $\mathbb{F}_{q^s}$ such that $\psi_{(ij)}(\alpha_1)=(\mathbf{e}_{\iota,1})_{(ij)},\ldots,\psi_{(ij)}(\alpha_{\Delta})=(\mathbf{e}_{\iota, \Delta})_{(ij)}$ and $\psi_{(ij)}(\beta_{l_1})=\mathbf{g}(\beta_{l_1})_{(ij)},\ldots,\psi_{(ij)}(\beta_{l_t})=\mathbf{g}(\beta_{l_t})_{(ij)}$. This implies that $I(\mathbf{g}(\beta_{l_1})_{(ij)},\ldots, \mathbf{g}(\beta_{l_t})_{(ij)};(\mathbf{e}_{\iota,1})_{(ij)},\ldots,(\mathbf{e}_{\iota, \Delta})_{(ij)})=0$ and by~\eqref{sec2},
the privacy property holds.

The property that responses from $r$ servers are enough for file retrieval trivially follows from the facts - that servers responses are values of polynomial $\phi$ over $\mathbb{F}_{q^s}$ of degree $r-1$ so that $\phi(\alpha_j)=x_j^{(\iota)}$ for $j\in[\Delta]$ and we can use a Lagrange interpolation formula to retrieve them. 

Let us prove the correctness of scheme~$\Pi_1$. It is clear that values $(\phi(\alpha_1),\ldots,\phi(\alpha_{\Delta}),\phi(\beta_1),\ldots,\phi(\beta_k))$ can be seen as a codeword of Reed-Solomon code 
\begin{align}\label{RS}
\textrm{RS}_{r}&(\Omega_{\alpha}\cup\Omega_{\beta})=\{(\phi(\alpha_1),\ldots,\phi(\alpha_{\Delta}),\phi(\beta_1),\notag \\
&\ldots,\phi(\beta_k))|\phi\in\mathbb{F}_{q^s}[\xi], \deg(\phi)<r\}. 
\end{align}

Dual of $\textrm{RS}_{r}$ is a Generalized-Reed Solomon code \cite{RSdecode} defined as
\begin{align}\label{GRS}
&\textrm{GRS}_{k+\Delta-r}(\Omega_{\alpha}\cup\Omega_{\beta})=\{u_1h(\alpha_1),\ldots,u_{\Delta}h(\alpha_{\Delta}),v_1h(\beta_1),\notag\\
&\ldots,v_{k}h(\beta_k))|h\in\mathbb{F}_{q^s}[\xi], \deg(h)<k+\Delta-r=k-t\},
\end{align}
where $u_{i}=\prod_{\ell\in[\Delta]\setminus\{i\}}(\alpha_{i}-\alpha_{\ell})^{-1}\times\prod_{j=1}^k(\alpha_{i}-\beta_j)^{-1}$ and $v_j=\prod_{\ell=1}^{\Delta}(\beta_j-\alpha_{\ell})^{-1}\times\prod_{\ell\in[k]\setminus\{j\}}(\beta_j-\beta_\ell)^{-1}$ for $i\in[\Delta]$ and $j\in[k]$.

As each $\alpha_j, j\in[\Delta]$ is a root of different monic irreducible polynomial $\tilde{f}_j$ of degree $s$ over $\mathbb{F}_q$ we have that
\begin{equation}
\tilde{f}_{j}(\alpha_{j})=0\;\;\;\tilde{f}_{j}(\alpha_i)\ne 0\;\;\;\textrm{for}\;\;i\in[\Delta], j\ne i
\end{equation}
\begin{equation}
\prod_{j\in[\Delta], j\ne i}\tilde{f}_{j}(\alpha_n)=0\;\;\;\textrm{for}\;\;n\in[\Delta], n\ne i.
\end{equation}

Let $\{\theta_1,\ldots,\theta_s\}$ be the basis of $\mathbb{F}_{q^s}$ over $\mathbb{F}_q$ and $\{\eta_1,\ldots,\eta_s\}$ is its trace-orthogonal basis. For each $\delta\in[s]$ and $i\in[\Delta]$, we can represent the element $u_{i}^{-1}\eta_{\delta}\prod_{\ell\in[\Delta], \ell\ne i}\tilde{f}_{\ell}^{-1}(\alpha_i)$ as the value of function $h_{i\delta}(\xi)\in\mathbb{F}_q[\xi]$ of degree less than $s$ at point $\alpha_{i}$. It is clear that  $\deg(h_{i\delta}\prod_{\ell\in[\Delta]\setminus\{i\}}\tilde{f}_{\ell})<\Delta s\leq k-t$ and hence such functions belong to the dual Generalized Reed-Solomon code~\eqref{GRS}. Also, we have that 
\begin{equation}
h_{i\delta}(\alpha_i)\prod_{\ell\in[\Delta]\setminus\{i\}}\tilde{f}_{\ell}(\alpha_{i})=u_{i}^{-1}\eta_{\delta} \end{equation}
and
\begin{equation}
h_{i\delta}(\alpha_n)\prod_{\ell\in[\Delta]\setminus\{i\}}\tilde{f}_{\ell}(\alpha_n)=0\;\;\;\textrm{for all}\;\;n\in[\Delta], n\ne i.
\end{equation}
Consequently,
\begin{align}
&\big(u_1h_{i\delta}(\alpha_1)\prod_{\ell\in[\Delta]\setminus\{i\}}\tilde{f}_{\ell}(\alpha_1),\ldots,u_{\Delta}h_{i\delta}(\alpha_{\Delta})\prod_{\ell\in[\Delta]\setminus\{i\}}\tilde{f}_{\ell}(\alpha_{\Delta}),\notag \\ 
&v_1h_{i\delta}(\beta_1)\prod_{\ell\in[\Delta]\setminus\{ i\}}\tilde{f}_{\ell}(\beta_1),\ldots, v_kh_{i\delta}(\beta_k)\prod_{\ell\in[\Delta]\setminus\{i\}}\tilde{f}_{\ell}(\beta_k)\big) \notag \\
&\cdot(\phi(\alpha_1),\ldots,\phi(\alpha_{\Delta}),\phi(\beta_1),\ldots,\phi(\beta_k))^T=0.
\end{align}

Utilizing the properties of function $h_{i\delta}(\xi)\in\mathbb{F}_q[\xi]$ we have
\begin{align}
&\eta_{\delta}\phi(\alpha_{i})+v_1h_{i\delta}\phi(\beta_1)\prod_{\ell\in[\Delta]\setminus\{i\}}\tilde{f}_{\ell}(\beta_1)+\ldots+\notag\\
&v_kh_{i\delta}\phi(\beta_k)\prod_{\ell\in[\Delta]\setminus\{i\}}\tilde{f}_{\ell}(\beta_k)=0  
\end{align}
and
\begin{equation}\label{rec11}
\eta_{\delta}\phi(\alpha_{i})=-\sum_{j=1}^{k}\left(v_jh_{i\delta}(\beta_j)\phi(\beta_j)\prod_{\ell\in[\Delta]\setminus\{i\}}\tilde{f}_{\ell}(\beta_j)\right).
\end{equation}

Applying trace-mapping function to both sides of equation~\eqref{rec11} and utilizing the facts that $h_{i\delta}(\xi)\in\mathbb{F}_q[\xi]$, $\tilde{f}_{\ell}(\xi)\in\mathbb{F}_q[\xi]$ and $\beta_j\in\mathbb{F}_q$ for all $i, \ell\in[\Delta]$, $\delta\in[s]$, $j\in[k]$ together with the linearity of trace-mapping function we obtain that
\begin{align}\label{rec21}
&\textrm{Tr}(\eta_{\delta}\phi(\alpha_{i}))=-\sum_{j=1}^k\textrm{Tr}\left(v_j\phi(\beta_j)h_{i\delta}(\beta_j)\prod_{\ell\in[\Delta]\setminus\{i\}}\tilde{f}_{\ell}(\beta_j)\right)=\notag \\
&-\sum_{j=1}^kh_{i\delta}(\beta_j)\left(\textrm{Tr}\left(v_j\phi(\beta_j)\right)\prod_{\ell\in[\Delta]\setminus\{i\}}\tilde{f}_{\ell}(\beta_j)\right)  
\end{align}

From the fact that $\{\theta_1,\ldots,\theta_s\} $ and $\{\eta_1,\ldots,\eta_s\}$ are trace-orthogonal bases of $\mathbb{F}_{q^s}$ over $\mathbb{F}_q$ it is clear (see, for example, \cite{Lidl}[Ch.~$2$]) that
\begin{equation}
x_{i}^{(\iota)}=\phi(\alpha_{i})=\sum_{\delta=1}^{s}\theta_{\delta}\textrm{Tr}(\eta_{\delta}\phi(\alpha_{i}))  \end{equation}
and hence all $\phi(\alpha_1),\ldots,\phi(\alpha_{\Delta})$ can be recovered by accessing $\textrm{Tr}(v_j\phi(\beta_j))$ from all involved servers $j=1,\ldots,k$.

The observations that each file consists of $\Delta$ elements of $\mathbb{F}_{q^s}$ for $s=\frac{k-t}{\Delta}$ and download rate is equal to $\frac{k-t}{k}$ finish the proof. 

\end{proof}



\section{Byzantine-Resistant PIR scheme}

Let us construct a $k$-server PIR scheme with $t$-colluding and $b$-byzantine servers by modifying the construction from Section~\ref{basicscheme}. For simplicity, we consider a non-universal case when we request responses from exactly $k\geq r$ servers, where $r$ is the recovery threshold\textcolor{black}{, and $(r-2b-t)$ divides $(k-2b-t)$. Here, we} employ the idea of~\cite{barg} to include error-correction capability in our PIR scheme.

\noindent\rule{9cm}{0.8pt}
\textit{Scheme\;$\Pi_2$:}\;$k$-server\;$t$-private\;$b$-byzantine\;resistant\;PIR
\noindent\rule{9cm}{0.4pt}

Let $t$, $r$, $k$ be positive integers satisfying $t<r-2b<k-2b\leq k\leq q$, $\Delta=r-2b-t$ and $\Delta|(k-2b-t)$.  Denote by $s\triangleq\frac{k-2b-t}{\Delta}$. Let $\Omega_{\alpha}=\{\alpha_1,\ldots,\alpha_{\Delta}\}\subset\mathbb{F}_{q^s}$, $\Omega_{\chi}=\{\chi_1,\ldots,\chi_t\}\subset\mathbb{F}_{q^s}$ and $\Omega_{\beta}=\{\beta_1,\ldots,\beta_k\}\subseteq\mathbb{F}_{q}$ be publicly known non-intersecting sets such that all elements of $\Omega_{\alpha}$ are roots of distinct monic irreducible polynomials of degree $s$ over $\mathbb{F}_q$.

Let us represent the database $\mathbf{x}$ with $m$ files as a $m\times\Delta$-array. Let the $(i,j)$-th entry of $\mathbf{x}$ be $x^{(i)}_j$. 
Then we set the file $f^{(i)}\triangleq[x^{(i)}_1,\ldots, x^{(i)}_\Delta]$.
Therefore,
\[
\mathbf{x} = \begin{bmatrix}
	x^{(1)}_1 & x^{(1)}_2 & \cdots & x^{(1)}_\Delta\\
	x^{(2)}_1 & x^{(2)}_2 & \cdots & x^{(2)}_\Delta\\
	\vdots & \vdots & \ddots &\vdots\\
	x^{(m)}_1 & x^{(m)}_2 & \cdots & x^{(m)}_\Delta
\end{bmatrix}	= 
\begin{bmatrix}
	f^{(1)} \\
	f^{(2)} \\
	\vdots \\
	f^{(m)}
\end{bmatrix}\,.
\] 
Then we define $\mathbf{e}_{(i,j)}$ be the $m\times \Delta$ indicator array for the database components. 
In other words,  the $(i,j)$th entry of $\mathbf{e}_{(i,j)}$ is one, while all other entries of $\mathbf{e}_{(i,j)}$ are zero. We replicate the database $\mathbf{x}$ on $k$ servers.
\begin{itemize}
\item \textit{Query generation algorithm:} To retrieve the file $\iota\in[m]$ user randomly generates $t$ $(m\times\Delta)$-arrays $\mathbf{r}^{(1)},\ldots,\mathbf{r}^{(t)}$ and draw a random degree-$(t+\Delta-1)$ curve 
    \begin{align}
\mathbf{g}(\xi) 
& \triangleq \sum_{j=1}^\Delta \prod_{\ell\in[\Delta]\setminus\{j\}}\left(\frac{\xi-\alpha_\ell}{\alpha_j-\alpha_\ell}\right)
							\prod_{\ell=1}^t \left(\frac{\xi-\chi_\ell}{\alpha_j-\chi_\ell}\right) \mathbf{e}_{(\textcolor{black}{\iota},j)} \notag \\
& + \sum_{h=1}^t\prod_{\ell=1}^\Delta \left(\frac{\xi-\alpha_\ell}{\chi_h-\alpha_\ell}\right)
\prod_{\ell\in[t]\setminus\{h\}} \left(\frac{\xi-\chi_\ell}{\chi_h-\chi_\ell}\right) \mathbf{r}^{(h)}\label{eq:query2}
\end{align}
that resides in $\mathbb{F}_{q^s}^{m\times\Delta}$ and passes through points $(\alpha_1,\mathbf{e}_{(\iota,1)}),\ldots, (\alpha_{\Delta},\mathbf{e}_{(\iota,\Delta)})$. Query to server $j\in[k]$ is $\mathbf{g}(\beta_j)$. We note that both $\mathbf{g}$ and $\mathbf{r}^{(h)}$ depend on retrieved index $\iota$, but we omit the subscript $\iota$ for readability.
\item \textit{Answer generation algorithm:} Upon receive the query $\mathbf{g}(\beta_j)$, server $j\in[k]$ computes the Frobenius inner product $\langle \mathbf{g}(\beta_j), \mathbf{x} \rangle$. We can observe that 
\begin{align*}
&\langle \mathbf{g}(\xi), \mathbf{x} \rangle 
 = \sum_{j=1}^\Delta \prod_{\ell\in [\Delta]\ne\{j\}}\left(\frac{\xi-\alpha_\ell}{\alpha_j-\alpha_\ell}\right)
\prod_{\ell=1}^t \left(\frac{\xi-\chi_\ell}{\alpha_j-\chi_\ell}\right) x^{(\iota)}_j \notag \\
& + \sum_{h=1}^t\prod_{\ell=1}^\Delta \left(\frac{\xi-\alpha_\ell}{\chi_h-\alpha_\ell}\right)
\prod_{\ell\in[t]\setminus\{h\}} \left(\frac{\xi-\chi_\ell}{\chi_h-\chi_\ell}\right) \langle\mathbf{r}^{(h)},\mathbf{x} \rangle,
\end{align*}
which is a polynomial in $\xi$ of degree $\Delta+t-1=r-2b-1$. 
We call this polynomial $\phi(\xi)$ and observe further that $\phi(\alpha_i)=x^{(\iota)}_i$ for $i\in[\Delta]$.
\begin{itemize}
    \item For retrieval from answers from $r$ servers, server $j$ responds with value of $\mathbf{a}_j=\phi(\beta_j)\in\FF_{q^s}$.
    \item For retrieval from answers from $k$ servers, server responds with 
    \begin{equation}
    \mathbf{a}_j=\textrm{Tr}(v_j\phi(\beta_j))\in\FF_q,   \end{equation}
    where 
    \begin{equation}
    v_j=\prod_{\ell=1}^{\Delta}(\beta_j-\alpha_{\ell})^{-1}\times\prod_{\ell\in[k]\setminus\{j\}}(\beta_j-\beta_\ell)^{-1}.
\end{equation}
\end{itemize}
\item \textit{File retrieval algorithm} 
\begin{itemize}
    \item For retrieval from answers from $r$ servers, the user applies any Reed-Solomon code decoding algorithm (see, for example, \cite{RSdecode}).
    \item For retrieval from answers from $k$ servers, the user decodes the vector
    \begin{equation}
    \left(\prod_{\ell=1}^{\Delta}\tilde{f}_{\ell}(\beta_1)\textrm{Tr}(v_1\phi(\beta_1)),\ldots,\prod_{\ell=1}^{\Delta}\tilde{f}_{\ell}(\beta_k)\textrm{Tr}(v_k\phi(\beta_k))\right),
    \end{equation}
    where $\tilde{f}_{\ell}(\xi)$ is the minimal polynomial of $\alpha_{\ell}$ over $\mathbb{F}_q$, as a codeword of Generalized Reed-Solomon code over $\mathbb{F}_q$ by any decoding algorithm (see, for example, \cite{RSdecode}) and extract the values 
    \begin{equation}
    \textrm{Tr}(v_j\phi(\beta_j)),\;\;\textrm{for}\;\;j\in[k].
    \end{equation}
    After it, user prepares a basis $\{\theta_1,\ldots,\theta_s\}$ for $\mathbb{F}_{q^s}$ over $\mathbb{F}_q$ and its trace-orthogonal basis $\{\eta_1,\ldots,\eta_s\}$. After it, the user chooses polynomials $h_{i\delta}\in\mathbb{F}_q[\xi]$ of degree less than $s$ for all $i\in[\Delta]$ and $\delta\in[s]$ so that
    \begin{equation}
    h_{i\delta} (\alpha_{i})=u_{i}^{-1}\eta_{\delta}\prod_{\ell\in[\Delta]\setminus\{i\}}\tilde{f}_{\ell}^{-1}(\alpha_{i}),
    \end{equation}
    and
    \begin{equation}
    u_{i}=\prod_{\ell\in[\Delta]\setminus\{i\}}(\alpha_{i}-\alpha_{\ell})^{-1}\times\prod_{j=1}^k(\alpha_{i}-\beta_j)^{-1}. 
    \end{equation}
    The user retrieves the file of interest by
    \begin{align}
    &x_{i}^{(\iota)}=\phi(\alpha_{i})=-\sum_{\delta=1}^s\theta_{\delta}\Bigg(\sum_{j=1}^kh_{i\delta}(\beta_j)\textrm{Tr}(v_j\phi(\beta_j))\cdot\notag \\
    &\prod_{\ell\in[\Delta]\setminus\{i\}}\tilde{f}_{\ell}(\beta_j)\Bigg),
    \end{align}
    for all $i\in[\Delta]$.
\end{itemize}
\end{itemize}
\noindent\rule{9cm}{0.4pt}
\begin{theorem}\label{theorem2}
Scheme~$\Pi_2$ is $k$-server $t$-private $b$-byzantine PIR over $\mathbb{F}_{q^s}$ with file size $(k-2b-t)\log(q)$ and recovery threshold $r$ that achieves the asymptotic capacity~\eqref{ascapacity} for any given $\Delta$ and $r$ so that
\begin{equation*}
t<r-2b<k-2b\leq q,\;\Delta=r-2b-t\;\;\textrm{and}\;\;\Delta|(k-2b-t),  
\end{equation*}
and $s=\frac{k-2b-t}{\Delta}$.
\end{theorem}
\begin{proof}
According to the definition of $k$-server $t$-private $b$-byzantine resistant PIR, we will prove privacy and correctness properties and show that responses from $r$ servers are enough for file retrieval in presence of up to $b$ incorrect responses. The proof of privacy coincides with privacy proof for Theorem~\ref{theorem1} and is omitted here. 

The property that responses from $r$ servers are enough for file retrieval follows from the fact that values $(\phi(\alpha_1),\ldots,\phi(\alpha_{\Delta}),\phi(\beta_{1}),\ldots,\phi(\beta_{k}))$ can be seen as a codeword of Reed-Solomon code
\begin{align}\label{RS1}
\textrm{RS}_{r-2b}&(\Omega_{\alpha}\cup\Omega_{\beta})=\{(\phi(\alpha_1),\ldots,\phi(\alpha_{\Delta}),\phi(\beta_{1}),\notag \\
&\ldots,\phi(\beta_{k}))|\phi\in\mathbb{F}_{q^s}[\xi], \deg(\phi)<r-2b\}. 
\end{align}
As a result, by values of polynomial $\phi$ in any $r$ points, we can correctly interpolate it in the presence of $b$ incorrect value utilizing any Reed-Solomon decoding algorithm (see, for example, \cite{RSdecode}). 

Let us prove the correctness of scheme $\Pi_2$. It is clear that dual of $\textrm{RS}_{r-2b}$ is a Generalized-Reed Solomon code \cite{RSdecode} defined as
\begin{align}\label{GRS2}
&\textrm{GRS}_{k+\Delta-(r-2b)}(\Omega_{\alpha}\cup\Omega_{\beta})=\{u_1h(\alpha_1),\ldots,u_{\Delta}h(\alpha_{\Delta}),v_1h(\beta_1),\notag\\
&\ldots,v_{k}h(\beta_k))|h\in\mathbb{F}_{q^s}[\xi], \deg(h)<k+\Delta-r+2b=k-t\},
\end{align}
where $u_{i}=\prod_{\ell\in[\Delta]\setminus\{i\}}(\alpha_{i}-\alpha_{\ell})^{-1}\times\prod_{j=1}^k(\alpha_{i}-\beta_j)^{-1}$ and $v_j=\prod_{\ell=1}^{\Delta}(\beta_j-\alpha_{\ell})^{-1}\times\prod_{\ell\in[k]\setminus\{j\}}(\beta_j-\beta_\ell)^{-1}$ for $i\in[\Delta]$ and $j\in[k]$.

As each $\alpha_j, j\in[\Delta]$ is a root of different monic irreducible polynomial $\tilde{f}_j$ of degree $s$ over $\mathbb{F}_q$ we have that
\begin{equation}
\tilde{f}_{j}(\alpha_{j})=0\;\;\;\tilde{f}_{j}(\alpha_i)\ne 0\;\;\;\textrm{for}\;\;i\in[\Delta], j\ne i
\end{equation}
\begin{equation}
\prod_{j\in[\Delta]\setminus\{i\}}\tilde{f}_{j}(\alpha_n)=0\;\;\;\textrm{for}\;\;n\in[\Delta], n\ne i.
\end{equation}

Let $\{\theta_1,\ldots,\theta_s\}$ be the basis of $\mathbb{F}_{q^s}$ over $\mathbb{F}_q$ and $\{\eta_1,\ldots,\eta_s\}$ is its trace-orthogonal basis. For each $\delta\in[s]$ and $i\in[\Delta]$, we can represent the element $u_{i}^{-1}\eta_{\delta}\prod_{\ell\in[\Delta]\setminus\{i\}}\tilde{f}_{\ell}^{-1}(\alpha_i)$ as the value of function $h_{i\delta}(\xi)\in\mathbb{F}_q[\xi]$ of degree less than $s$ at point $\alpha_{i}$. It is clear that  $\deg(h_{i\delta}\prod_{\ell\in[\Delta]\setminus\{i\}}\tilde{f}_{\ell})<\Delta s<k-t$ and hence such functions belong to the dual Generalized Reed-Solomon code~\eqref{GRS2}. Also, we have that
\begin{equation}
h_{i\delta}(\alpha_i)\prod_{\ell\in[\Delta]\setminus\{i\}}\tilde{f}_{\ell}(\alpha_{i})=u_{i}^{-1}\eta_{\delta} \end{equation}
and
\begin{equation}
h_{i\delta}(\alpha_n)\prod_{\ell\in[\Delta]\setminus\{i\}}\tilde{f}_{\ell}(\alpha_n)=0\;\;\;\textrm{for all}\;\;n\in[\Delta], n\ne i.
\end{equation}
Consequently,
\begin{align}
&\big(u_1h_{i\delta}(\alpha_1)\prod_{\ell\in[\Delta]\setminus\{i\}}\tilde{f}_{\ell}(\alpha_1),\ldots,u_{\Delta}h_{i\delta}(\alpha_{\Delta})\prod_{\ell\in[\Delta]\setminus\{i\} }\tilde{f}_{\ell}(\alpha_{\Delta}),\notag \\ 
&v_1h_{i\delta}(\beta_1)\prod_{\ell\in[\Delta]\setminus\{i\}}\tilde{f}_{\ell}(\beta_1),\ldots, v_kh_{i\delta}(\beta_k)\prod_{\ell\in[\Delta]\setminus\{i\}}\tilde{f}_{\ell}(\beta_k)\big) \notag \\
&\cdot(\phi(\alpha_1),\ldots,\phi(\alpha_{\Delta}),\phi(\beta_1),\ldots,\phi(\beta_k))^T=0.
\end{align}

Utilizing the properties of function $h_{i\delta}(\xi)\in\mathbb{F}_q[\xi]$ we can write down that 
\begin{align}
&\eta_{\delta}\phi(\alpha_{i})+v_1h_{i\delta}\phi(\beta_1)\prod_{\ell\in[\Delta]\setminus\{i\}}\tilde{f}_{\ell}(\beta_1)+\ldots+\notag\\
&v_kh_{i\delta}\phi(\beta_k)\prod_{\ell\in[\Delta]\setminus\{i\}}\tilde{f}_{\ell}(\beta_k)=0  
\end{align}
and
\begin{equation}\label{rec111}
\eta_{\delta}\phi(\alpha_{i})=-\sum_{j=1}^{k}\left(v_jh_{i\delta}(\beta_j)\phi(\beta_j)\prod_{\ell\in[\Delta]\setminus\{i\}}\tilde{f}_{\ell}(\beta_j)\right).
\end{equation}

Applying trace-mapping function to both sides of equation~\eqref{rec111} and utilizing the facts that $h_{i\delta}(\xi)\in\mathbb{F}_q[\xi]$, $\tilde{f}_{\ell}(\xi)\in\mathbb{F}_q[\xi]$ and $\beta_j\in\mathbb{F}_q$ for all $i, \ell\in[\Delta]$, $\delta\in[s]$, $j\in[k]$ together with the linearity of trace-mapping function we obtain that
\begin{align}\label{rec22}
&\textrm{Tr}(\eta_{\delta}\phi(\alpha_{i}))=-\sum_{j=1}^k\textrm{Tr}\left(v_j\phi(\beta_j)h_{i\delta}(\beta_j)\prod_{\ell\in[\Delta]\setminus\{i\} }\tilde{f}_{\ell}(\beta_j)\right)=\notag \\
&-\sum_{j=1}^kh_{i\delta}(\beta_j)\left(\textrm{Tr}\left(v_j\phi(\beta_j)\right)\prod_{\ell\in[\Delta]\setminus\{i\}}\tilde{f}_{\ell}(\beta_j)\right)  
\end{align}

From the fact that $\{\theta_1,\ldots,\theta_s\} $ and $\{\eta_1,\ldots,\eta_s\}$ are trace-orthogonal bases of $\mathbb{F}_{q^s}$ over $\mathbb{F}_q$ it is clear (see, for example, \cite{Lidl}[Ch.~$2$]) that
\begin{equation}
x_{i}^{(\iota)}=\phi(\alpha_{i})=\sum_{\delta=1}^{s}\theta_{\delta}\textrm{Tr}(\eta_{\delta}\phi(\alpha_{i}))  \end{equation}
an hence all $\phi(\alpha_1),\ldots,\phi(\alpha_{\Delta})$ can be recovered by accessing $\textrm{Tr}(v_j\phi(\beta_j))$ from all involved servers $j=1,\ldots,k$.

Let us show that we can correctly recover $\phi(\alpha_1),\ldots,\phi(\alpha_{\Delta})$ even in case of at most $b$ incorrect values of $\textrm{Tr}(v_j\phi(\beta_j))$. Following the ideas from \cite{barg}, let us replace the functions $h_{i\delta}$ in derivations above by functions 
\begin{equation}
    \tilde{h}_{e}(\xi)=\xi^e\prod_{\ell\in[\Delta]}\tilde{f}_{\textcolor{black}{\ell}}(\xi)\in\mathbb{F}_q[\xi].
\end{equation}

It is clear that $\tilde{h}_{e}(\alpha_{i})=0$ for all $i\in[\Delta]$ and $\deg(h_e(\xi))<\Delta s+e$. Hence, for all $e<2b$, we have that $\deg(h_e(\xi))<k-t$ and, as a result, these functions belong to the dual Generalized Reed-Solomon code~\eqref{GRS2}. Consequently, 
\begin{align}\label{repeq3}
&(u_1\tilde{h}_e(\alpha_1),\ldots,u_{\Delta}\tilde{h}_e(\alpha_{\Delta}),v_1\tilde{h}_e(\beta_1),\ldots, v_k\tilde{h}_e(\beta_k))\cdot \notag \\ 
&(\phi(\alpha_1),\ldots,\phi(\alpha_{\Delta}),\phi(\beta_1),\ldots,\phi(\beta_k))^T=\notag \\
&v_1\tilde{h}_e(\beta_1)\phi(\beta_1)+\ldots+v_k\tilde{h}_e(\beta_k)\phi(\beta_k)=0.
\end{align}

As $\tilde{h}_e(\xi)\in\mathbb{F}_q[\xi]$ and $\{\beta_1,\ldots,\beta_k\}\subseteq\mathbb{F}_q$, applying the trace-mapping function to both sides of equation~\eqref{repeq3} and utilizing its linearity, we have 
\begin{align}
&\beta_1^e\prod_{\ell\in[\Delta]}\tilde{f}_{\ell}(\beta_1)\textrm{Tr}(v_1\phi(\beta_1))+\ldots\notag \\ 
&+\beta_k^e\prod_{\ell\in[\Delta]}\tilde{f}_{\ell}(\beta_k)\textrm{Tr}(v_k\phi(\beta_k))=0,
\end{align}
where $e=0,1,\ldots,2b-1$, and $\beta_j$, $\prod_{\ell\in[\Delta]}\tilde{f}_{\ell}(\beta_k)$, $\textrm{Tr}(v_j\phi(\beta_j))$ belong to $\mathbb{F}_q$ for all $j\in[k]$. Hence, elements $\prod_{\ell\in[\Delta]}\tilde{f}_{\ell}(\beta_k)\textrm{Tr}(v_k\phi(\beta_k))$ form a codeword of Generalized Reed-Solomon code over $\mathbb{F}_q$ of length $k$ and dimension $k-2b$. This code can correct up to $b$ errors, and user can retrieve the correct values of $\textrm{Tr}(v_jf(\beta_j))$ from server responses for all $j\in[k]$ \cite{RSdecode}.

The observations that each file consists of $\Delta$ elements of $\mathbb{F}_{q^s}$ for $s=\frac{k-t-2b}{\Delta}$ and download rate is equal to $\frac{k-t-2b}{k}$ finish the proof.

\end{proof}
\begin{table*}[ht!]
\caption{Parameters of PIR scheme with optimal download rate. All parameters are measured in bits. }\label{table}
\scalebox{0.96}
{
\begin{tabular}{|l|l|l|l|l|}
\hline
                     & $\Pi_1$                                       & $\Pi_2$                                             & $\mathbb{A}_1$       & $\mathbb{A}_2$             \\ \hline
File size            & $l(k-t)\log(q)$                               & $l(k-2b-t)\log(q)$                                  & $l(r-t)(k-t)\log(q)$ & $l(r-2b-t)(k-2b-t)\log(q)$ \\ \hline
Field                & $\mathbb{F}_{q^s}$, where $s=\frac{k-t}{r-t}$ & $\mathbb{F}_{q^s}$, where $s=\frac{k-2b-t}{r-2b-t}$ & $\mathbb{F}_q$       & $\mathbb{F}_q$             \\ \hline
Download cost        & $lk\log(q)$                                    & $lk\log(q)$                                          & $lk(r-t)\log(q)$      & $lk(r-2b-t)\log(q)$         \\ \hline
Download rate        & $1-\frac{t}{k}$                               & $1-\frac{2b+t}{k}$                                  & $1-\frac{t}{k}$      & $1-\frac{2b+t}{k}$         \\ \hline
Capacity             & $1-\frac{t}{k}$                               & $1-\frac{2b+t}{k}$                                  & $1-\frac{t}{k}$      & $1-\frac{2b+t}{k}$         \\ \hline
Byzantine-resistance & $0$                                           & $b\ne0$                                             & $0$                  & $b\ne0$                    \\ \hline
\end{tabular}
}
\end{table*}

\section{Lower bound on the file size}
For any given base field $\mathbb{F}_q$ of size $q\geq k$, the size of file $\Delta s\log_2q$ bits depends on the parameter $\Delta s$. In this section, following the derivations of~\cite[Section~V]{communicationefficientSS}, we describe the class of byzantine-resistant PIR schemes for an asymptotically large number of files, and then show that the size of the file in our scheme is optimal. 

\begin{definition}[Balanced byzantine-resistant PIR]
The byzantine-resistant PIR scheme is balanced if the client downloads a single element of the same subfield $\mathbb{F}_{q^R}$ from each involved server.
\end{definition}

\begin{definition}[Rate optimal byzantine-resistant PIR]
Byzantine-resistant PIR scheme is rate optimal if \textcolor{black}{for any $i$th file $f^{(i)}$} $H(f^{(i)})=(r-2b-t)s=s\Delta$.   
\end{definition}

After introducing the necessary definitions, we can formulate the main theorem of this section. 

\begin{theorem}
For a balanced rate-optimal byzantine-resistant PIR scheme that achieves the minimum download rate for a large enough number of files, the following hold:
\begin{equation}
(r-2b-t)|(k-2b-t)
\end{equation}
\begin{equation}
s\Delta\geq k-2b-t
\end{equation}
\end{theorem}
\begin{proof}
As the considered scheme is rate-optimal, 
$\Delta=r-2b-t$. Let us consider the scenario in which the client downloads a single element of $\mathbb{F}_{q^R}$ from each of $k$ servers, and the scheme achieves the asymptotic capacity. It follows that
\begin{equation}
k\log_{\textcolor{black}{q}}q^R=\frac{k}{k-2b-t}\Delta s.
\end{equation}

As a result $\Delta s=R(k-2b-t)$. Since $s$ is a positive integer and $\mathbb{F}_{q^R}$ is a subfield of $\mathbb{F}_{q^s}$, then $R$ divides $s$ and hence theorem statement follows.
\end{proof}
\begin{corollary}
The scheme $\Pi_2$ is a balanced rate-optimal byzantine-resistant PIR with optimal file size. 
\end{corollary}

\section{Comparison}

In this section, we give a comparison of PIR schemes $\Pi_1$ and $\Pi_2$ with Staircase-PIR from~\cite{BitarRouayheb}. We denote the latter as scheme $\mathbb{A}_1$ and formulate its parameters in form of the following theorem.
\begin{theorem}
Scheme $\mathbb{A}_1$ is $k$-server $t$-private PIR over $\mathbb{F}_{q}$ with file size $(r-t)(k-t)\log(q)$ and recovery threshold $r$ that achieves the asymptotic capacity~\eqref{ascapacity} for $b=0$ and any given r so that
\begin{equation*}
t<r<k\leq q.
\end{equation*}
\end{theorem}

Scheme $\mathbb{A}_1$ assumes that servers are honest-but-curious and provide correct answers. To add the $b$-byzantine resistance, we can utilize error-correction capabilities of underlined staircase codes in the same way as it was done in \cite{bj}. We present the parameters of resulted scheme $\mathbb{A}_2$ in the following theorem.
\begin{theorem}
Scheme $\mathbb{A}_2$ is $k$-server $t$-private $b$-byzantine resistant PIR over $\mathbb{F}_{q}$ with file size $(r-2b-t)(k-2b-t)\log(q)$ and recovery threshold $r$ that achieves the asymptotic capacity~\eqref{ascapacity} for any given $r$ so that
\begin{equation*}
t<r-2b<k-2b<k\leq q.
\end{equation*}
\end{theorem}

To justify ignoring the upload cost, we repeat each scheme $l$ times and summarize the parameters in Table~\ref{table}. We note that our schemes work over the extended field, while schemes $\mathcal{A}_1$ and $\mathbb{A}_2$ work over the base field. Nevertheless, in $\mathbb{A}_1$ and $\mathbb{A}_2$,  each component of the file consists of multiple field symbols that result in a big file size and retrieval delay.

\section{Conclusion}
We considered the problem of designing a Private Information Retrieval scheme resistant to the adversarial behavior of servers. We focused on download cost minimization and proposed a \textcolor{black}{non-universal} capacity-achieving scheme with a small file size \textcolor{black}{for asymptotically large number of files of fixed size}. We also formally proved that such a file size is optimal solving the problem pointed out by Banawan and Ulukus in \cite{BPIR1}. \textcolor{black}{Extending the
proposed framework to the universal case and finite number of files are interesting open problems.}

{\section*{Acknowledgements.}
This research/project is supported by the National Research Foundation, Singapore under its Strategic Capability Research Centres Funding Initiative, Singapore Ministry of Education Academic Research Fund Tier 2 Grants MOE2019-T2-2-083 and MOE-T2EP20121-0007, and the ARC grants DE180100768 and DP200100731. Any opinions, findings and conclusions or recommendations expressed in this material are those of the author(s) and do not reflect the views of National Research Foundation, Singapore.}

\balance
\printbibliography
\end{document}